\documentclass[11pt]{article}
\usepackage{geometry}
\usepackage{amsmath, amssymb, amsthm}
\usepackage{graphicx}
\usepackage{bm}
\usepackage[linesnumbered,ruled]{algorithm2e}
\SetKwRepeat{Do}{do}{while}
\usepackage{fullpage}
\newtheorem{theorem}{Theorem}
\newtheorem{lemma}{Lemma}

\newtheorem{claim}{Claim}

\title{Non-uniformly Stable Matchings\thanks{%
This work was supported by JST ERATO Grant Number JPMJER2301, Japan.}}
\author{Naoyuki Kamiyama}
\date{Institute of Mathematics for Industry \\ Kyushu University, Fukuoka, Japan \\
{\ttfamily kamiyama@imi.kyushu-u.ac.jp}}

\begin{document}

\maketitle

\begin{abstract}
Super-stability and strong stability are properties of a matching in 
the stable matching problem with ties. 
In this paper, we introduce  
a common generalization of super-stability and 
strong stability, which 
we call non-uniform stability. 
First, we prove that we can determine the existence of a 
non-uniformly stable matching in polynomial time.
Next, we give a polyhedral characterization of 
the set of non-uniformly stable 
matchings. 
Finally, we prove that 
the set of non-uniformly stable matchings forms 
a distributive lattice. 
\end{abstract} 

\section{Introduction}

In this paper, we consider a problem of finding a matching satisfying 
some property between two disjoint groups of agents. 
In particular, we are interested in the setting where 
each agent has a preference over agents in the other group. 
Stability, which was introduced by Gale and Shapley~\cite{GaleS62}, 
is one of the most fundamental properties in this setting.
Informally speaking, 
a matching is said to be stable if 
there does not exist an unmatched pair of agents having  
incentive to deviate from the current matching. 
Gale and Shapley~\cite{GaleS62} proved that if 
the preference of an agent is strict (i.e., a strict total order), then
there always exists a stable matching and we can find a stable matching 
in polynomial time. 

Irving~\cite{Irving94} considered 
a variant of the stable matching problem where 
ties are allowed in the preferences of agents. 
For this setting, 
Irving~\cite{Irving94} introduced the following three properties of a matching
(see, e.g., \cite{IwamaM08} and 
\cite[Chapter~3]{Manlove13} for a survey of the stable matching problem 
with ties).  
The first property is weak stability. 
This stability guarantees that 
there does not exist an unmatched pair of agents such that 
both prefer the other agent to the current partner. 
Irving~\cite{Irving94} proved 
that there always exists a weakly stable matching, and 
a weakly stable matching can be found
in polynomial time by using the algorithm 
of Gale and Shapley~\cite{GaleS62} with tie-breaking.
The second one is strong stability.
This stability guarantees that 
there does not exist an unmatched pair of agents such that
the other agent is not worse than the current partner 
for both agents, 
and at least one of the agents prefers the other agent 
to the current partner. 
The third one is super-stability.
This stability guarantees that 
there does not exist an unmatched pair of agents such that
the other agent is not worse than the current partner 
for both agents. 
It is known that a super-stable matching and a strongly stable matching 
may not exist~\cite{Irving94}. 
Thus, the problem of checking 
the existence of a matching satisfying these properties
is 
one of the fundamental algorithmic challenges for 
super-stability and strong stability.
For example, 
Irving~\cite{Irving94} 
proposed 
polynomial-time algorithms for checking the existence of 
a super-stable matching and 
a strongly stable matching
(see also \cite{Manlove99}). 

Similar results have been obtained for 
both super-stability and strong stability. 
Since their definitions are different, 
the same results are separately 
proved for these properties 
by using the same techniques.  
The main motivation of this paper is the following 
theoretical question.
Can we prove the same results  
for super-stability and strong stability at the same time? 
In order to solve this question, 
in this paper, we first introduce a common generalization of 
super-stability and 
strong stability, which 
we call non-uniform stability. 
Then, in Section~\ref{section:algorithm}, we prove that we can check the existence of a 
non-uniformly stable matching in polynomial time.
Our algorithm is a careful unification of the 
algorithms proposed by Irving~\cite{Irving94} 
for super-stability and strong stability. 
Next, in Section~\ref{section:polytope}, we give a polyhedral characterization of 
the set of non-uniformly stable matchings (i.e., 
a linear inequality description of the convex hull of the set of 
characteristic vectors of non-uniformly stable matchings).
For the setting where the preferences are strict, 
polyhedral characterizations of 
the set of stable matchings were given in, e.g., \cite{Fleiner03,Rothblum92,VandeVate89}. 
For the setting where ties in the preferences are allowed, 
Hu and Garg~\cite{HuG21} gave 
a polyhedral characterization of 
the set of super-stable matchings, and 
Kunysz~\cite{Kunysz18}
gave a polyhedral characterization of 
the set of strongly stable matchings
(see also \cite{HuG21,Juarez23}).
Our proposed characterization is a natural 
unification of those given by 
Hu and Garg~\cite{HuG21}
and 
Kunysz~\cite{Kunysz18}.
Our proof is based on the proofs of 
Hu and Garg~\cite{HuG21} for 
super-stability and strong stability.
Finally, in Section~\ref{section:structure}, 
we first prove a structure of the set of vertices covered by 
a non-uniformly stable matching, which 
is a common generalization of the same results 
for 
super-stability~\cite[Lemma~4.1]{Manlove99}
and 
strong stability~\cite[Lemma~3.1]{Manlove99}. 
Then we prove that the set of non-uniformly stable matchings 
forms a distributive lattice.
This result is a common generalization
of the distributive lattice structures of 
stable matchings with strict preferences 
(Knuth~\cite{Knuth97} attributes this result to Conway), 
super-stable matchings~\cite{Spieker95}, and 
strongly stable matchings~\cite{Manlove02}.

\subsection{Related work} 

In the one-to-one setting, 
Irving~\cite{Irving94} 
proposed polynomial-time algorithms for 
the super-stable matching problem and 
the strongly stable matching problem 
(see also \cite{Manlove99}). 
Kunysz, Paluch, and Ghosal~\cite{KunyszPG16}
considered a 
characterization of the set of all strongly stable matchings. 
Kunysz~\cite{Kunysz18} considered 
the weighted version of the strongly stable matching
problem. 
In the many-to-one setting, 
Irving, Manlove, and Scott~\cite{IrvingMS00} proposed 
a polynomial-time algorithm for 
the super-stable matching problem.
Furthermore, 
Irving, Manlove, and Scott~\cite{IrvingMS03} and 
Kavitha, Mehlhorn, Michail, and Paluch~\cite{KavithaMMP07}
proposed polynomial-time algorithms for the 
strongly stable matching
problem. 
In the many-to-many setting, 
Scott~\cite{Scott05} considered 
the super-stable matching problem, 
and the papers~\cite{ChenG10,Kunysz19,Malhotra04} considered 
the strongly stable matching problem.
Olaosebikan and Manlove~\cite{OlaosebikanM20,OlaosebikanM22}
considered super-stability and strong stability in the 
student-project allocation problem. 
Furthermore, 
Kamiyama~\cite{Kamiyama22,Kamiyama22b} considered 
strong stability and super-stability 
under matroid constraints.
Strong stability and super-stability 
with master 
lists were 
considered in, e.g., \cite{IrvingMS08,Kamiyama15,Kamiyama19,OMalley07}.

\section{Preliminaries} 

Let $\mathbb{R}_+$ denote the set of non-negative real numbers. 
For each finite set $U$, each vector $x \in \mathbb{R}_+^U$, 
and each subset $W \subseteq U$, we define 
$x(W) := \sum_{u \in W}x(u)$. 
For each positive integer $z$, we define $[z] := \{1,2,\dots,z\}$. 
Define $[0] := \emptyset$. 

\subsection{Setting} 

Throughout this paper, we are given  
a finite simple undirected bipartite graph $G = (V,E)$ such 
that the vertex set $V$ of $G$ is partitioned into 
$V_1$ and $V_2$, and every edge in the edge set $E$ of $G$ connects a vertex 
in $V_1$ and a vertex in $V_2$. 
In addition, we 
are given subsets $E_1,E_2 \subseteq E$ such that 
$E_1 \cap E_2 = \emptyset$ and $E_1 \cup E_2 = E$.
In this paper, we do not distinguish between an edge $e \in E$ 
and the set consisting 
of the two end vertices of $e$. 
For each edge $e \in E$, if we write $e = (v,w)$, then 
this means that $e \cap V_1 = \{v\}$ and 
$e \cap V_2 = \{w\}$. 

For each subset $F \subseteq E$ and each vertex $v \in V$, 
let $F(v)$ denote the set of edges 
$e \in F$ 
such that $v \in e$. 
Thus, $E(v)$ denotes the set of edges in $E$ incident to $v$. 
A subset $\mu \subseteq E$ is called a \emph{matching in $G$} 
if
$|\mu(v)| \le 1$ for every vertex 
$v \in V$.
For each 
matching $\mu$ in $G$ and 
each vertex $v \in V$ such that $|\mu(v)| = 1$, 
we do not distinguish between 
$\mu(v)$ and the element in $\mu(v)$. 
For each subset $F \subseteq E$, 
a matching $\mu$ in $G$ is called a \emph{matching in $F$} if 
$\mu \subseteq F$. 
Furthermore, for each subset $F \subseteq E$, 
a matching $\mu$ in $F$ is called a \emph{maximum-size matching in $F$} if 
$|\mu| \ge |\sigma|$ for every matching $\sigma$ 
in $F$. 

For each vertex $v \in V$, we are given 
a transitive and complete
binary relation 
$\succsim_v$ on $E(v) \cup \{\emptyset\}$ 
such that
$e \succsim_v \emptyset$ and 
$\emptyset \not\succsim_v e$
for every edge $e \in E(v)$. 
(In this paper, the completeness of a binary relation means 
that, 
for every pair of elements $e,f \in E(v) \cup \{\emptyset\}$,
at least one of $e \succsim_v f$ and $f \succsim_v e$ holds.)
For each vertex $v \in V$ and each pair of elements $e,f \in E(v) \cup \{\emptyset\}$, 
if $e \succsim_v f$ and $f \not\succsim_v e$
(resp.\ $e \succsim_v f$ and $f \succsim_v e$), then 
we write $e \succ_v f$ (resp.\ $e \sim_v f$). 
Intuitively speaking, 
$e \succ_v f$ means that 
$v$ prefers $e$ to $f$, and 
$e \sim_v f$ means that 
$v$ is indifferent between $e$ and $f$.

\subsection{Non-uniform stability} 

Let $\mu$ be a matching in $G$.
For each edge $e \in E \setminus \mu$, 
we say that 
$e$ \emph{weakly blocks} $\mu$ if  
$e \succsim_v \mu(v)$ for every vertex $v \in e$.
Furthermore, 
for each edge $e \in E \setminus \mu$, 
we say that 
$e$ \emph{strongly blocks} $\mu$ if 
$e \succsim_v \mu(v)$ 
for every vertex $v \in e$
and 
there exists a vertex $w \in e$ such that 
$e \succ_w \mu(w)$.
Then $\mu$ is said to be 
\emph{non-uniformly stable} if 
the following two conditions are satisfied. 
\begin{description}
\item[(S1)]
Any edge $e \in E_1 \setminus \mu$ does not weakly 
block $\mu$. 
\item[(S2)] 
Any edge $e \in E_2 \setminus \mu$ does not strongly block $\mu$. 
\end{description}
If $E_1 = E$, then non-uniform stability is equivalent to 
super-stability.
Furthermore, if $E_2 = E$, then non-uniform stability is equivalent to 
strong stability.
Thus, non-uniform stability is a common generalization of 
super-stability and 
strong stability. 
For simplicity, we may say that an edge $e \in E \setminus \mu$ \emph{blocks} $\mu$ 
if the following condition is satisfied. 
If $e \in E_1$, then 
$e$ weakly blocks $\mu$. 
Otherwise (i.e., $e \in E_2$), $e$ strongly blocks $\mu$. 

\subsection{Notation} 

Let $F$ be a subset of $E$. Let $i$ be an integer in $\{1,2\}$. 
Then we define the characteristic vector $\chi_F \in \{0,1\}^E$ of $F$ 
by 
$\chi_F(e) := 1$ for each edge $e \in F$ and 
$\chi_F(e) := 0$ for each edge $e \in E \setminus F$. 
For each subset $X \subseteq V$, 
we define $X(F)$ as the set of vertices $v \in X$ 
such that $F(v) \neq \emptyset$. 
Thus, for each integer $i \in \{1,2\}$, 
$V_i(F)$ denotes the set of vertices $v \in V_i$
such that $F(v) \neq \emptyset$. 
For each subset $X \subseteq V_i$, 
we define $F(X) := \bigcup_{v \in X}F(v)$.
For each subset $X \subseteq V_i$, 
we define $\Gamma_F(X)$ as the set of 
vertices $v \in V_{3-i}$ 
such that $F(X) \cap E(v) \neq \emptyset$.
That is, $\Gamma_F(X)$ is the set of vertices in $V_{3-i}$ adjacent to 
a vertex in $X$ via an edges in $F$.  
For each vertex $v \in V$, we use   
$\Gamma_F(v)$ instead of $\Gamma_F(\{v\})$. 
Define the real-valued function $\rho_{i,F}$ on $2^{V_i(F)}$ by 
$\rho_{i,F}(X) := |\Gamma_F(X)| - |X|$
for each subset $X \subseteq V_i(F)$. 
It is not difficult to see that,
for every pair of 
subsets $X,Y \subseteq V_i(F)$,  
\begin{equation*}
\rho_{i,F}(X) + \rho_{i,F}(Y) \ge \rho_{i,F}(X \cup Y) + \rho_{i,F}(X \cap Y),
\end{equation*}
i.e., $\rho_{i,F}$ is submodular.
Then it is known that 
there exists the unique (inclusion-wise)  
minimal minimizer of $\rho_{i,F}$, and we can find it in polynomial 
time (see, e.g., \cite[Note~10.12]{Murota03}). 
In addition, it is known that
there exists a matching 
$\mu$ in $F$ such that 
$|\mu| = |V_i(F)|$ if and only if, 
for every subset $X \subseteq V_i(F)$, 
$|X| \le |\Gamma_F(X)|$ holds~\cite{Hall35} (see also \cite[Theorem~16.7]{Schrijver02}). 

Let $v$ be a vertex in $V$.
Then, for each edge $e \in E(v)$ and 
each symbol $\odot \in \{\succ_v,\succsim_v, \sim_v\}$, 
we define 
$E[\mathop{\odot} e]$ as the set of edges $f \in E(v)$ 
such that $f \odot e$. 
Furthermore, 
for each edge $e \in E(v)$ and 
each symbol $\odot \in \{\succ_v,\succsim_v\}$, 
we define 
$E[e \mathop{\odot}]$ as the set of edges $f \in E(v)$ 
such that $e \odot f$. 

Let $F$ be a subset of $E$. 
A sequence $(v_1,v_2,\dots,v_k)$ of distinct vertices in 
$V$ is called a \emph{path in $F$} if 
$\{v_i,v_{i+1}\} \in F$ for every integer $i \in [k-1]$. 
A sequence $(v_1,v_2,\dots,v_{k+1})$ of vertices in 
$V$ is called a \emph{cycle in $F$} if 
$(v_1,v_2,\dots,v_k)$ is a path in $F$,  
$v_1 = v_{k+1}$, and 
$\{v_k,v_1\} \in F$. 

\section{Algorithm} 
\label{section:algorithm}

In this section, we propose a polynomial-time algorithm 
for the problem of checking the existence of a non-uniformly stable matching, and 
finding it if it exists.  

For each vertex $v \in V_1$ and each subset $F \subseteq E$, 
we define ${\rm Ch}_v(F)$ as the set of edges $e \in F(v)$ such that 
$e \succsim_v f$ for every edge $f \in F(v)$.  
In other words, for each vertex $v \in V_1$, 
${\rm Ch}_v(F)$ is the set of the most preferable edges in $F(v)$ for $v$. 
For each vertex $v \in V_2$ and each subset 
$F \subseteq E$, we define ${\rm Ch}_v(F)$ as the 
output of Algorithm~\ref{alg:ch_2}.
For each vertex $v \in V_2$, 
${\rm Ch}_v(F)$ is some subset of the most preferable edges in $F(v)$ for $v$.

\begin{algorithm}[ht]
\KwIn{a vertex $v \in V_2$ and a subset $F \subseteq E$}
Define $B := \{e \in F(v) \mid  
\mbox{$e \succsim_v f$ for every edge $f \in F(v)$}\}$.\\
\uIf{$B \cap E_1 = \emptyset$}
{
  Output $B$ and halt. 
}
\uElseIf{$|B \cap E_1| = 1$}
{
  Output $B \cap E_1$ and halt. 
}
\Else
{
  Output $\emptyset$ and halt. 
}
\caption{Algorithm for defining ${\rm Ch}_v(F)$ for a vertex $v \in V_2$}
\label{alg:ch_2}
\end{algorithm}

For each subset $F \subseteq E$, 
we define 
\begin{equation*}
{\rm Ch}_1(F) := 
\bigcup_{v \in V_1}{\rm Ch}_v(F), \ \ \ 
{\rm Ch}_2(F) := 
\bigcup_{v \in V_2}{\rm Ch}_v(F), \ \ \ 
{\rm Ch}(F) := {\rm Ch}_2({\rm Ch}_1(F)). 
\end{equation*}
For each matching $\mu$ in $G$, 
we define ${\sf block}(\mu)$ as the set of edges $e = (v,w) \in E \setminus \mu$ 
such that 
$\mu(w) \neq \emptyset$ and 
$e$ blocks $\mu$. 

Our algorithm is described in Algorithm~\ref{alg:main}. 

\begin{algorithm}[ht]
Set $t := 0$. Define ${\sf R}_0 := \emptyset$.\\
\Do{${\sf R}_t \neq {\sf P}_{t,i_t}$}
{
  Set $t := t + 1$ and $i := 0$. Define ${\sf P}_{t,0} := {\sf R}_{t-1}$.\\
  \Do{${\sf P}_{t,i} \neq {\sf P}_{t,i-1}$}
  {
    Set $i := i+1$.\\
    Define $L_{t,i} := {\rm Ch}(E \setminus {\sf P}_{t,i-1})$ and 
    ${\sf Q}_{t,i} := {\rm Ch}_1(E \setminus {\sf P}_{t,i-1}) \setminus L_{t,i}$.\\  
    \uIf{${\sf Q}_{t,i} \neq \emptyset$}
    {
         Define ${\sf P}_{t,i} := {\sf P}_{t,i-1} \cup {\sf Q}_{t,i}$. 
    }    
    \Else
    {
        Find a maximum-size matching $\mu_{t,i}$ in $L_{t,i}$.\\
        \uIf{$|\mu_{t,i}| < |V_1(E \setminus {\sf P}_{t,i-1})|$}
        {
          Find the minimal minimizer $N_{t,i}$ of $\rho_{1,L_{t,i}}$.\\
          Define ${\sf P}_{t,i} := {\sf P}_{t,i-1} \cup L_{t,i}(N_{t,i})$. 
        }
        \Else
        {
          Define ${\sf P}_{t,i} := {\sf P}_{t,i-1}$.  
        }
    }
  }
    Define $i_t := i$.\\
    \uIf{${\sf P}_{t,i_t} \cap {\sf block}(\mu_{t,i_t}) \neq \emptyset$}
    {
        Define $e_t = (v_t,w_t)$ as an arbitrary edge in ${\sf P}_{t,i_t} 
        \cap {\sf block}(\mu_{t,i_t})$.\\
        Define ${\sf R}_t := {\sf P}_{t,i_t} \cup \{\mu_{t,i_t}(w_t)\}$. 
    }
    \Else
    {
        Define ${\sf R}_t := {\sf P}_{t,i_t}$. 
    }
}
Define $k := t$ and $\mu_{\rm o} := \mu_{k,i_k}$.\\
\uIf{there exists an edge 
$e_{\rm R} = (v_{\rm R},w_{\rm R}) \in {\sf R}_k \cup {\rm Ch}(E \setminus {\sf R}_k)$ such that 
$\mu_{\rm o}(w_{\rm R}) = \emptyset$}
{
   Output {\bf No} and halt. 
}
\Else
{
   Output $\mu_{\rm o}$ and halt. 
}
\caption{Proposed algorithm}
\label{alg:main}
\end{algorithm}

First, we prove that 
Algorithm~\ref{alg:main} is well-defined.
To this end, it is sufficient to 
prove that, 
in Step~20, $\mu_{t,i_t}$ is well-defined.
If ${\sf Q}_{t,i_t} \neq \emptyset$, then 
${\sf P}_{t,i_t-1} \neq {\sf P}_{t,i_t}$. 
However, 
this contradicts the fact that Algorithm~\ref{alg:main} 
proceeds to Step~19. 
Thus, ${\sf Q}_{t,i_t} = \emptyset$. 
In this case, $\mu_{t,i_t}$ is defined 
in Step~10.

In the course of Algorithm~\ref{alg:main}, 
${\sf P}_{t,i-1} \subseteq {\sf P}_{t,i}$ and 
${\sf R}_{t-1} \subseteq {\sf R}_t$. 
Furthermore, we can find $\mu_{t,i}$ in polynomial time
by using, e.g., the algorithm in \cite{HopcroftK73}, 
and 
$N_{t,i}$ can be found in polynomial time
(see, e.g., \cite[Note~10.12]{Murota03}). 
Thus, Algorithm~\ref{alg:main} is a polynomial-time
algorithm. 

In the rest of this section, we prove the correctness of 
Algorithm~\ref{alg:main}. 

First, we prove that if Algorithm~\ref{alg:main}
outputs $\mu_{\rm o}$, then 
$\mu_{\rm o}$ is a non-uniformly stable matching in $G$. 
To this end, we need the following lemmas. 

\begin{lemma} \label{lemma:algorithm_n_non_empty}
In Step~12 of Algorithm~\ref{alg:main}, 
we have $L_{t,i}(N_{t,i}) \neq \emptyset$. 
\end{lemma} 
\begin{proof}
Since $N_{t,i} \subseteq V_1(L_{t,i})$, 
it suffices to prove that 
$N_{t,i} \neq \emptyset$. 
Since $V_1(E \setminus {\sf P}_{t,i-1}) = V_1(L_{t,i})$
follows from ${\sf Q}_{t,i}=\emptyset$,
we have $|\mu_{t,i}| < |V_1(L_{t,i})|$. 
Thus, since $\mu_{t,i}$ is a maximum-size matching in $L_{t,i}$,
there exists a subset 
$X \subseteq V_1(L_{t,i})$ such that 
$|X| > |\Gamma_{L_{t,i}}(X)|$, i.e., $\rho_{1,L_{t,i}}(X) < 0$. 
Since $N_{t,i}$ is a minimizer of $\rho_{1,L_{t,i}}$, 
this implies that 
$\rho_{1,L_{t,i}}(N_{t,i}) < 0$.
Thus, $N_{t,i} \neq \emptyset$.
\end{proof} 

Notice that the definition of Algorithm~\ref{alg:main}
implies that 
${\sf P}_{k,i_k-1} = {\sf P}_{k,i_k} = {\sf R}_k$. 
In what follows, we use this fact. 

\begin{lemma} \label{lemma:alg_property} 
In Step~27 of Algorithm~\ref{alg:main}, 
$\mu_{\rm o}(v) \in {\rm Ch}(E \setminus {\sf R}_k)$
for every vertex $v \in V_1(E \setminus {\sf R}_k)$. 
\end{lemma}
\begin{proof}
It follows from Lemma~\ref{lemma:algorithm_n_non_empty} and 
the condition in Step~11
that $|\mu_{\rm o}| \ge |V_1(E \setminus {\sf R}_{k})|$. 
Thus, since $\mu_{\rm o} \subseteq L_{k,i_k} \subseteq E \setminus {\sf R}_k$, 
$\mu_{\rm o}(v) \neq \emptyset$ holds
for every vertex $v \in V_1(E \setminus {\sf R}_k)$. 
This implies that since 
$\mu_{\rm o} \subseteq L_{k,i_k} = {\rm Ch}(E \setminus {\sf R}_k)$, 
$\mu_{\rm o}(v) \in {\rm Ch}(E \setminus {\sf R}_{k})$ 
for every vertex $v \in V_1(E \setminus {\sf R}_k)$. 
\end{proof} 

\begin{lemma} \label{lemma:alg_yes}
If Algorithm~\ref{alg:main} outputs $\mu_{\rm o}$, then 
$\mu_{\rm o}$ is a non-uniformly stable matching in $G$. 
\end{lemma}
\begin{proof}
Assume that Algorithm~\ref{alg:main} outputs $\mu_{\rm o}$.
Since $\mu_{\rm o}$ is clearly a matching in $G$, 
we prove that $\mu_{\rm o}$ is non-uniformly stable. 
Assume that there exists an edge $e = (v,w) \in E \setminus \mu_{\rm o}$
blocking $\mu_{\rm o}$. 

We first consider the case where $\mu_{\rm o}(w) = \emptyset$. 
In this case, if $e \in {\sf R}_k$, then 
Algorithm~\ref{alg:main} 
outputs ${\bf No}$ in 
Step~29.  
This is a contradiction. 
Thus, we can assume that $e \notin {\sf R}_k$. 
Then since $e$ blocks $\mu_{\rm o}$,
we have 
$e \succsim_v \mu_{\rm o}(v)$. 
Since $e \in E \setminus {\sf R}_k$, 
we have $v \in V_1(E \setminus {\sf R}_k)$. 
Thus, Lemma~\ref{lemma:alg_property} implies that  
$\mu_{\rm o}(v) \in {\rm Ch}(E \setminus {\sf R}_k)$.  
Since $e \succsim_v \mu_{\rm o}(v)$
and $e \in E \setminus {\sf R}_k$, 
we have 
$e \in {\rm Ch}_1(E \setminus {\sf R}_k)$. 
Since ${\sf Q}_{k,i_k} = \emptyset$, 
this implies that 
$e \in {\rm Ch}(E \setminus {\sf R}_k)$.
This implies that 
Algorithm~\ref{alg:main} 
outputs ${\bf No}$ in 
Step~29.  
This is a contradiction. 

Next, we consider the case where 
$\mu_{\rm o}(w) \neq \emptyset$. 
In this case, $e \in {\sf block}(\mu_{\rm o})$. 
Thus, if $e \in {\sf P}_{k,i_k}$, 
then ${\sf R}_k \neq {\sf P}_{k,i_k}$.
This is a contradiction. 
Thus, $e \notin {\sf P}_{k,i_k} = {\sf R}_k$, i.e., 
$e \in E \setminus {\sf R}_k$.  
This implies that since $\mu_{\rm o} \subseteq {\rm Ch}(E \setminus {\sf R}_k)$
and $e$ blocks $\mu_{\rm o}$, 
we have $e \sim_v \mu_{\rm o}(v)$. 
Thus, $e \in {\rm Ch}_1(E \setminus {\sf R}_k)$. 
Since ${\sf Q}_{k,i_k} = \emptyset$, 
$e \in {\rm Ch}(E \setminus {\sf R}_k)$.
Thus, $e \sim_w \mu_{\rm o}(w)$. 
Since $e$ blocks $\mu_{\rm o}$, 
$e \in E_1$. 
Thus, 
the definition of ${\rm Ch}_w(\cdot)$ implies that 
$\mu_{\rm o}(w) \notin L_{k,i_k}$. 
This contradicts the fact that $\mu_{\rm o} \subseteq L_{k,i_k}$.
This completes the proof. 
\end{proof} 

Next, we prove that
if Algorithm~\ref{alg:main} outputs {\bf No}, then 
there does not exist 
a non-uniformly stable matching in $G$. 
The following lemma plays an important role in the proof of this part. 
We prove this lemma in the next subsection.

\begin{lemma} \label{lemma:alg_intersection} 
For every non-uniformly stable matching $\sigma$ in $G$, 
we have $\sigma \cap {\sf R}_k = \emptyset$. 
\end{lemma} 

Furthermore, we need the following lemma. 

\begin{lemma} \label{lemma:alg_r}
In the course of Algorithm~\ref{alg:main}, 
for every vertex $v \in V_1$, every edge 
$e \in {\sf R}_t(v)$, and every edge $f \in E(v) \setminus {\sf R}_t$, 
we have $e \succsim_v f$. 
\end{lemma} 
\begin{proof}
This lemma follows from the fact that 
${\sf P}_{t,i} \setminus {\sf P}_{t,i-1} 
\subseteq {\rm Ch}_1(E \setminus {\sf P}_{t,i-1})$ for every integer $i \in [i_t]$
and $\mu_{t,i_t} \subseteq L_{t,i_t} \subseteq {\rm Ch}_1(E \setminus {\sf P}_{t,i_t})$. 
\end{proof}

\begin{lemma} \label{lemma:alg_no} 
If Algorithm~\ref{alg:main} outputs {\bf No}, then 
there does not exist 
a non-uniformly stable matching in $G$. 
\end{lemma}
\begin{proof} 
Assume that Algorithm~\ref{alg:main}
outputs {\bf No} and 
there exists a non-uniformly stable matching $\sigma$ in $G$. 
Then there exists an edge 
$e_{\rm R} = (v_{\rm R},w_{\rm R}) \in {\sf R}_k \cup {\rm Ch}(E \setminus {\sf R}_k)$ 
such that $\mu_{\rm o}(w_{\rm R}) = \emptyset$. 
We define $\mu^+ := \mu_{\rm o} \cup \{e_{\rm R}\}$. 
Then since $\mu_{\rm o} \subseteq {\rm Ch}(E \setminus {\sf R}_k)$, 
$\mu^+ \subseteq {\sf R}_k \cup {\rm Ch}(E \setminus {\sf R}_k)$. 
Lemma~\ref{lemma:alg_intersection} implies 
that $\sigma \subseteq E \setminus {\sf R}_k$. 
Thus,
since $\mu_{\rm o}(w_{\rm R}) = \emptyset$, 
Lemma~\ref{lemma:alg_property} implies that 
$|V_2(\sigma)| \le |V_2(\mu_{\rm o})| < |V_2(\mu^+)|$. 
This implies that there exists an edge $e = (v,w) \in \mu^+ \setminus \sigma$ such that 
$\mu^+(w) \neq \emptyset$ and $\sigma(w) = \emptyset$. 
Since $e \in \mu^+ \subseteq {\sf R}_k \cup {\rm Ch}(E \setminus {\sf R}_k)$ and 
$\sigma \subseteq E \setminus {\sf R}_k$, 
Lemma~\ref{lemma:alg_r} implies that 
$e \succsim_v \sigma(v)$. 
However, this contradicts the fact that 
$\sigma$ is non-uniformly stable. 
\end{proof}

\begin{theorem} \label{theorem:algorithm}
If Algorithm~\ref{alg:main} outputs $\mu_{\rm o}$, then 
$\mu_{\rm o}$ is a non-uniformly stable matching in $G$. 
If Algorithm~\ref{alg:main} outputs {\bf No}, then 
there does not exist a non-uniformly stable matching in $G$. 
\end{theorem} 
\begin{proof}
This theorem follows from Lemmas~\ref{lemma:alg_yes} and 
\ref{lemma:alg_no}. 
\end{proof}

\subsection{Proof of Lemma~\ref{lemma:alg_intersection}}

In this subsection, we prove Lemma~\ref{lemma:alg_intersection}. 

An edge $e \in E$ is called a \emph{bad edge} if 
$e \in {\sf R}_k$ and 
there exists a non-uniformly stable matching $\sigma$ in $G$ 
such that $e \in \sigma$. 
If there does not exist a bad edge in $E$, then 
the proof is done. 
Thus, we assume that there exists a bad edge in $E$.
Then we define $\Delta$ as the set of integers $t \in [k]$ such that 
${\sf R}_t \setminus {\sf R}_{t-1}$ contains a bad edge in $E$. 
Furthermore, we define $z$ as the minimum integer in $\Delta$. 
We divide the proof into the following two cases. 

\begin{description}
\item[Case~1.]
There exists an integer $i \in [i_z]$ such that 
${\sf P}_{z,i} \setminus {\sf P}_{z,i-1}$ contains a bad edge in $E$.
\item[Case~2.]
${\sf P}_{z,i_z} \setminus {\sf R}_{z-1}$ does not contain a bad edge in $E$.
\end{description}

{\bf Case~1.} 
Define $q$ as the minimum integer in $[i_z]$ such that 
${\sf P}_{z,q} \setminus {\sf P}_{z,q-1}$ contains a bad edge in $E$.
Let $f = (v,w)$ be a bad edge in $E$ that is contained in 
${\sf P}_{z,q} \setminus {\sf P}_{z,q-1}$. 
Then there exists a non-uniformly stable matching $\sigma$ in $G$ 
such that $f \in \sigma$. 
If $\sigma \cap {\sf P}_{z,q-1} \neq \emptyset$, then 
this contradicts the minimality of $q$. Thus, 
$\sigma \subseteq E \setminus {\sf P}_{z,q-1}$. 

First, we consider the case where $f \in {\sf Q}_{z,q}$. 
The definition of Algorithm~\ref{alg:ch_2} implies that 
there exists an edge $g = (u,w) \in {\rm Ch}_1(E \setminus {\sf P}_{z,q-1})$ 
satisfying one of the following conditions.
\begin{itemize}
\item
$g \succ_w f$. 
\item
$g \sim_w f$, $g \neq f$, and $g \in E_1$.
\end{itemize}
Furthermore, since $\sigma \subseteq E \setminus {\sf P}_{z,q-1}$, 
$g \succsim_u \sigma(u)$.
Thus, $g$ blocks $\sigma$. 
This is a contradiction. 

Next, we consider the case where $f \in L_{z,q}(N_{z,q})$. 
In this case, 
${\sf Q}_{z,q} = \emptyset$. 
For simplicity, 
we 
define $L := L_{z,q}$ and 
$N := N_{z,q}$. 
Since ${\sf Q}_{z,q} = \emptyset$, 
$L = {\rm Ch}_1(E \setminus {\sf P}_{z,q-1})$.

\begin{claim} \label{claim:alg_special_vertex}
There exists a vertex $v \in N$ satisfying the following 
conditions.
\begin{itemize}
\item
$\sigma(v) \notin L$.
\item
There exists a vertex $w \in \Gamma_L(v)$ such that 
$\sigma(w) \in L$.
\end{itemize} 
\end{claim}
\begin{proof}
Define $T$ as the set of vertices $v \in N$ 
such that $\sigma(v) \in L$.
Then the existence of $f$ implies that 
$T \neq \emptyset$. 
Thus, $N\setminus T \subsetneq N$.
Since $\sigma$ is a matching in $G$, 
$|\Gamma_{\sigma}(T)| = |T|$. 
In order to 
prove this claim, 
it is sufficient to prove that 
there exists 
a vertex $v \in N \setminus T$ such that 
$\Gamma_{L}(v) \cap \Gamma_{\sigma}(T) \neq \emptyset$. 
Assume that there exists such a vertex $v \in N \setminus T$.
Let $w$ be a vertex in $\Gamma_{L}(v) \cap \Gamma_{\sigma}(T)$.
Then since $w \in \Gamma_{\sigma}(T)$, 
there exists a vertex $u \in T$ such that $w \in \Gamma_{\sigma}(u)$. 
Since $\sigma$ is a matching in $G$, 
$\Gamma_{\sigma}(u) = \{w\}$.
Thus, $\sigma(w) = \sigma(u) \in L$.  
If $\Gamma_{L}(v) \cap \Gamma_{\sigma}(T) = \emptyset$ for every vertex  
$v \in N \setminus T$, 
then $\Gamma_L(N \setminus T) \subseteq \Gamma_L(N) \setminus \Gamma_{\sigma}(T)$. 
Since $\sigma(u) \in L$ for every vertex $u \in T$, 
$\Gamma_{\sigma}(T) \subseteq \Gamma_L(N)$.
Thus, 
\begin{equation*}
\begin{split}
& |\Gamma_L(N \setminus T)| - |N \setminus T|
\le 
|\Gamma_L(N) \setminus \Gamma_{\sigma}(T)| - |N| + |T|\\
& =   
|\Gamma_L(N)| - |\Gamma_{\sigma}(T)| - |N| + |T|
=  
|\Gamma_L(N)| - |N|.
\end{split}
\end{equation*}
However, this contradicts the fact that 
$N$ is the minimal minimizer of $\rho_{1,L}$. 
\end{proof} 

Let $v$ be a vertex in $N$ satisfying the 
conditions in Claim~\ref{claim:alg_special_vertex}. 
Furthermore, let $w$ be 
a vertex in
$\Gamma_{L}(v)$ such that 
$\sigma(w) \in L$.
Let $e$ be the edge in $L$ such that 
$e = (v,w)$. 
Since 
$e, \sigma(w) \in L$, 
we have $e \sim_w \sigma(w)$.  
Since 
${\sf P}_{z,q-1}$ does not contain a bad edge in $E$, 
$\sigma(v) \in E \setminus {\sf P}_{z,q-1}$. 

\begin{claim} \label{claim:alg_l}
$e \succ_v \sigma(v)$. 
\end{claim} 
\begin{proof}
Recall that $L = {\rm Ch}_1(E \setminus {\sf P}_{z,q-1})$.
If $\sigma(v) \succsim_v e$, then since 
$e \in L$ and $\sigma(v) \in E \setminus {\sf P}_{z,q-1}$, 
we have $\sigma(v) \in L$.
However, this contradicts the fact that $\sigma(v) \notin L$. 
\end{proof} 

Claim~\ref{claim:alg_l} implies that 
$e$ blocks $\sigma$.
This contradicts the fact that $\sigma$ is
non-uniformly stable. 

{\bf Case~2.} 
In this case, 
$\mu_{z,i_z}(w_z)$ is a bad edge in $E$.
Thus, there exists a non-uniformly stable matching $\sigma$ in $G$ 
such that $\mu_{z,i_z}(w_z) \in \sigma$. 
Since ${\sf P}_{z,i_z}$ does not contain a bad edge in $E$, 
we have 
$\sigma \cap {\sf P}_{z,i_z} = \emptyset$ and $e_z \notin \sigma$. 
Since $e_z$ blocks $\mu_{z,i_z}$, 
if $\mu_{z,i_z}(v_z) \succsim_{v_z} \sigma(v_z)$, then 
$e_z$ blocks $\sigma$. 
Since 
\begin{equation*}
\mu_{z,i_z} \subseteq L_{z,i_z} = {\rm Ch}(E \setminus {\sf P}_{z,i_z-1}) 
= {\rm Ch}(E \setminus {\sf P}_{z,i_z}) 
\end{equation*}
and $\sigma \subseteq E \setminus {\sf P}_{z,i_z}$, 
we have 
$\mu_{z,i_z}(v_z) \succsim_{v_z} \sigma(v_z)$. 
This completes the proof. 

\section{Polyhedral Characterization}
\label{section:polytope}  

Define ${\bf P}$ as the convex 
hull of $\{\chi_{\mu} \mid \mbox{$\mu$ is a non-uniformly stable matching in $G$}\}$. 
Furthermore, we define ${\bf S}$ as 
the set of vectors $x \in \mathbb{R}_+^E$ satisfying the following conditions.
\begin{equation} \label{eq_1:constraint} 
x(E(v)) \le 1 
 \ \ \ \mbox{($\forall v \in V$)}.
\end{equation}
\begin{equation} \label{eq_2:constraint} 
\displaystyle{x(e) + \sum_{v \in e} x(E[\succ_v e]) \ge 1} 
 \ \ \ \mbox{($\forall e \in E_1$)}.
\end{equation}
\begin{equation} \label{eq_3:constraint} 
\displaystyle{x(E[\sim_v e]) + \sum_{w \in e}x(E[\succ_w e]) \ge 1} 
 \ \ \ \mbox{($\forall e \in E_2$, $\forall v \in e$)}. 
\end{equation}
Hu and Garg~\cite{HuG21} proved that 
if $E_1 = E$, then ${\bf P} = {\bf S}$.
Furthermore, 
Kunysz~\cite{Kunysz18} proved that 
if $E_2 = E$, then 
${\bf P} = {\bf S}$ (see also \cite{HuG21,Juarez23}).
In this section, we prove the following theorem.

\begin{theorem} \label{theorem:polytope} 
${\bf P} = {\bf S}$.
\end{theorem} 

Theorem~\ref{theorem:polytope} follows from the following two lemmas. 

\begin{lemma} \label{lemma:polytope_contain} 
For every non-uniformly stable matching $\mu$ in $G$, 
$\chi_{\mu}$ satisfies \eqref{eq_1:constraint}, 
\eqref{eq_2:constraint}, and \eqref{eq_3:constraint}. 
\end{lemma}
\begin{proof}
Since $\mu$ is a matching in $G$, 
$\chi_{\mu}$ satisfies \eqref{eq_1:constraint}.

Next, we prove that $\chi_{\mu}$ satisfies \eqref{eq_2:constraint}.
Let $e$ be an edge in $E_1$.  
If $e \in \mu$, then since $\chi_{\mu}(e) = 1$, 
$\chi_{\mu}$
satisfies \eqref{eq_2:constraint}. 
Assume that $e \notin \mu$. 
Then since $\mu$ is non-uniformly stable,
there exists a vertex $v \in e$ such that 
$\mu(v) \succ_v e$.
This implies that 
$\chi_{\mu}(E[\succ_v e]) \ge 1$. 
Thus, $\chi_{\mu}$ satisfies \eqref{eq_2:constraint}.

Finally, we prove that $\chi_{\mu}$ satisfies \eqref{eq_3:constraint}. 
Let $e$ be an edge in $E_2$.  
If $e \in \mu$, then 
$\chi_{\mu}(E[\sim_v e]) \ge 1$ for every vertex $v \in e$.
Thus, $\chi_{\mu}$ satisfies 
\eqref{eq_3:constraint}.
Assume that $e \notin \mu$. 
Then, in this case, since $\mu$ is non-uniformly stable, 
one of the following 
statements holds.
(i) There exists a vertex $w \in e$ such that 
$\mu(w) \succ_w e$. 
(ii) $\mu(v) \sim_v e$ for every vertex $v \in e$. 
If (i) holds, then $\chi_{\mu}(E[\succ_w e]) \ge 1$.
If (ii) holds, then 
$\chi_{\mu}(E[\sim_v e]) \ge 1$ for every vertex $v \in e$.
Thus, $\chi_{\mu}$ satisfies 
\eqref{eq_3:constraint}.
\end{proof} 

\begin{lemma} \label{lemma:polytope_extreme_point} 
For every extreme point $x$ of ${\bf S}$, 
there exists a non-uniformly stable matching $\mu$ in $G$ such that 
$\chi_{\mu} = x$. 
\end{lemma}

We give the proof of Lemma~\ref{lemma:polytope_extreme_point} in the next subsection.

\begin{proof}[Proof of Theorem~\ref{theorem:polytope}]
Since ${\bf S}$ is bounded,
${\bf S}$ is a polytope (see, e.g., \cite[Corollary 3.14.]{ConfortiCZ14}).  
Thus, 
${\bf S}$ is the convex hull of the extreme points of ${\bf S}$
(see, e.g., \cite[Theorem~3.37]{ConfortiCZ14}).
This implies that 
this theorem follows from 
Lemmas~\ref{lemma:polytope_contain} and \ref{lemma:polytope_extreme_point}.
\end{proof} 

\subsection{Proof of Lemma~\ref{lemma:polytope_extreme_point}} 

In this proof, we fix an extreme point $x$ of ${\bf S}$. 
Then we prove that 
there exists a non-uniformly stable matching $\mu$ in $G$ such that 
$\chi_{\mu} = x$. 
The following lemma implies that it is sufficient to 
prove that $x \in \{0,1\}^E$. 

\begin{lemma} \label{lemma:polytope_zero_one}
For every vector $y \in {\bf S} \cap \{0,1\}^E$,  
there exists a non-uniformly stable matching $\mu$ in $G$
such that $\chi_{\mu} = y$. 
\end{lemma} 
\begin{proof}
Let $y$ be a vector in ${\bf S} \cap \{0,1\}^E$. 
Define $\mu$ as the set of edges $e \in E$ 
such that $y(e) = 1$. 
Clearly, $\chi_{\mu} = y$. Thus, it suffices 
to prove that $\mu$ is a non-uniformly stable matching in $G$.
Since \eqref{eq_1:constraint} implies that 
$\mu$ is a matching in $G$, 
we prove that 
$\mu$ is non-uniformly stable. 

Let $e$ be an edge in $E \setminus \mu$. 
First, we assume that $e \in E_1$. 
Since $e \notin \mu$, $y(e) = 0$.  
Thus, \eqref{eq_2:constraint} implies that 
there exist a vertex $v \in e$ and 
an edge $f \in E(v)$ such that 
$y(f) = 1$ and $f \succ_v e$. 
This implies that
$e$ does not block $\mu$. 
Next, we assume that $e \in E_2$. 
Then 
\eqref{eq_3:constraint} implies that 
at least 
one of the following statements holds.
(i) There exist a vertex $w \in e$ and 
an edge $f \in E(w)$ such that 
$y(f) = 1$ and $f \succ_w e$. 
(ii) For every vertex $v \in e$, 
there exists an edge $f \in E(v)$ 
such that 
$y(f) = 1$ and $f \sim_v e$. 
In any case, 
$e$ does not block $\mu$. 
This completes the proof. 
\end{proof} 

In the rest of this proof, we prove that $x \in \{0,1\}^E$. 

Define ${\sf supp}$ as the set of 
edges $e \in E$ such that $x(e) > 0$. 
For each vertex $v \in V$ such that 
${\sf supp}(v) \neq \emptyset$, we define 
${\sf T}(v)$ (resp.\ ${\sf B}(v)$)
as the set of edges $e \in {\sf supp}(v)$
such that 
$e \succsim_v f$ 
(resp.\ $f \succsim_v e$) 
for every edge $f \in {\sf supp}(v)$.
For each integer $i \in \{1,2\}$, 
we define $S_i$ as the set of vertices $v \in V_i$ such that 
${\sf supp}(v) \neq \emptyset$. 
For each integer $i \in \{1,2\}$, 
we define $F_i := \bigcup_{v \in S_i}{\sf T}(v)$.
It should be noted that 
$S_i = V_i(F_i)$. 

\begin{lemma} \label{lemma:polytope_property}
Let $v$ and $e = \{v,w\}$ be a vertex in $V$ and 
an edge in ${\sf T}(v)$, respectively.
\begin{description}
\item[(E1)]
If $e \in E_1$, then  
$x(E(w)) = 1$, and 
${\sf B}(w) = \{e\}$.
\item[(E2)]
If $e \in E_2$, then  
$x(E(w)) = 1$,
$e \in {\sf B}(w)$, and 
$x(E[\sim_v e]) \ge x(E[\sim_w e])$. 
\end{description}
\end{lemma}
\begin{proof}
{\bf (E1)}
Since $e \in {\sf T}(v)$, 
$x(f) = 0$ for every edge 
$f \in E[\succ_v e]$. 
Thus, \eqref{eq_2:constraint} implies that 
\begin{equation*}
\begin{split}
1 & \le x(e) + x(E[\succ_v e]) + x(E[\succ_w e])  
= x(e) + x(E[\succ_w e]) \\
& = x(E(w)) - x(E[e \succ_w]) - x(E[\sim_w e] \setminus \{e\}) 
 \le 1 - x(E[e \succ_w]) - x(E[\sim_w e] \setminus \{e\}) \le 1. 
\end{split}
\end{equation*}
This implies that
(i) $x(E(w)) = 1$, 
(ii) $x(f) = 0$ for every edge $f \in E[e \succ_w]$, and 
(iii) $x(f) = 0$ for every edge $f \in E[\sim_w e] \setminus \{e\}$.
Thus, (ii) and (iii) 
imply that ${\sf B}(w) = \{e\}$. 

{\bf (E2)}
Since $e \in {\sf T}(v)$, 
$x(f) = 0$ for every edge 
$f \in E[\succ_v e]$.
Thus, \eqref{eq_3:constraint} for $w$ implies that 
\begin{equation*}
\begin{split}
1 & \le x(E[\sim_w e]) + x(E[\succ_v e]) + x(E[\succ_w e])  
=  x(E[\sim_w e]) + x(E[\succ_w e]) \\
& = x(E(w)) - x(E[e \succ_w]) 
 \le 1 - x(E[e \succ_w]) \le 1.
\end{split}
\end{equation*}
Thus, we have (i) $x(E(w)) = 1$, and 
(ii) $x(f) = 0$ for every edge $f \in E[e \succ_w]$.
Thus, (ii) implies that $e \in {\sf B}(w)$. 
Furthermore, \eqref{eq_3:constraint} for $v$ implies that 
\begin{equation*}
\begin{split}
1 & \le x(E[\sim_v e]) + x(E[\succ_v e]) + x(E[\succ_w e]) 
= x(E[\sim_v e]) + x(E[\succ_w e]) \\
& = x(E[\sim_v e]) + x(E(w)) - x(E[e \succsim_w]) 
\le x(E[\sim_v e]) + 1 - x(E[e \succsim_w]) \\
& = x(E[\sim_v e]) + 1 - x(E[\sim_w e]).
\end{split}
\end{equation*}
This implies that 
$x(E[\sim_v e]) \ge x(E[\sim_w e])$. 
This completes the proof. 
\end{proof} 

\begin{lemma} \label{lemma:polytope_matching}
For every integer $i \in \{1,2\}$, 
there exists a matching $\mu_i$ in $F_i$ such that 
$|\mu_i| = |S_i|$.
\end{lemma}
\begin{proof}
Let $i$ be an integer in $\{1,2\}$. 
Assume that there does not exist
a matching $\mu_i$ in $F_i$ such that 
$|\mu_i| = |S_i|$.
Since $S_i = V_i(F_i)$, 
there exists a subset $N \subseteq V_i(F_i)$ such that 
$|N| > |\Gamma_{F_i}(N)|$. 
Let $N$ be the minimal minimizer of $\rho_{i,F_i}$. 
For simplicity, we define $M := \Gamma_{F_i}(N)$. 
\begin{claim}
There exists a matching $\sigma$ in $F_i(N)$ 
such that $|\sigma| = |M|$. 
\end{claim}
\begin{proof}
Assume that 
there does not exist a matching $\sigma$ in $F_i(N)$ 
such that $|\sigma| = |M|$. 
In this case, 
since $M = V_{3-i}(F_i(N))$, 
there exists a subset $X \subseteq M$ such that 
$|X| > |\Gamma_{F_i(N)}(X)|$. 
Since $\Gamma_{F_i}(v) \cap X = \emptyset$ for every vertex 
$v \in N \setminus \Gamma_{F_i(N)}(X)$, 
$\Gamma_{F_i}(N \setminus \Gamma_{F_i(N)}(X)) \subseteq M \setminus X$.
Thus, 
\begin{equation*}
\begin{split}
& |M| - |N| - |\Gamma_{F_i}(N \setminus \Gamma_{F_i(N)}(X))| 
+ |N \setminus \Gamma_{F_i(N)}(X)|\\
& \ge
|M| - |N| - |M| + |X| + |N| - |\Gamma_{F_i(N)}(X)| 
= |X| - |\Gamma_{F_i(N)}(X)| > 0. 
\end{split}
\end{equation*}
This contradicts the fact that $N$ is a minimizer of $\rho_{i,F_i}$. 
\end{proof} 

Let $\sigma$ be a matching in $F_i(N)$ 
such that $|\sigma| = |M|$.
Then Lemma~\ref{lemma:polytope_property} implies that 
\begin{equation*}
\sum_{v \in N(\sigma)}x(E[\sim_v \sigma(v)])
\ge 
\sum_{w \in M}x(E[\sim_w \sigma(w)]). 
\end{equation*}
Furthermore, 
since Lemma~\ref{lemma:polytope_property} implies that 
${\sf T}(v) \subseteq \bigcup_{w \in M}{\sf B}(w)$ 
for every vertex $v \in N(\sigma)$,
\begin{equation*}
\sum_{w \in M}x(E[\sim_w \sigma(w)])
= 
\sum_{w \in M}x({\sf B}(w))
\ge
\sum_{v \in N(\sigma)}x({\sf T}(v))
= 
\sum_{v \in N(\sigma)}x(E[\sim_v \sigma(v)]).
\end{equation*}
Thus, we have 
\begin{equation*}
\sum_{w \in M}x({\sf B}(w))
=  
\sum_{v \in N(\sigma)}x({\sf T}(v)).
\end{equation*}
Since $|N| > |M|$, 
there exists an edge $f \in F_i(N \setminus N(\sigma))$.
Furthermore, Lemma~\ref{lemma:polytope_property} implies that 
$f \in {\sf B}(w)$ for some vertex $w \in M$. 
Thus, 
\begin{equation*}
\sum_{w \in M}x({\sf B}(w))
=  
\sum_{v \in N(\sigma)}x({\sf T}(v))
<  
x(f) + \sum_{v \in N(\sigma)}x({\sf T}(v)) 
\le 
\sum_{w \in M}x({\sf B}(v)). 
\end{equation*}
However, this is a contradiction. 
\end{proof} 

In what follows, 
for each integer $i \in \{1,2\}$, 
let $\mu_i$ be a matching in $F_i$ such that 
$|\mu_i| = |S_i|$.

\begin{lemma} \label{lemma:polytope_size}
$|S_1| = |S_2|$. 
\end{lemma}
\begin{proof}
For every integer $i \in \{1,2\}$, 
the existence of $\mu_i$ implies that 
$|S_i| \le |S_{3-i}|$. 
\end{proof} 

Lemma~\ref{lemma:polytope_size}
implies that
$\mu_1(v) \neq \emptyset$ if and only if
$\mu_2(v) \neq \emptyset$
for every vertex $v \in V$.
Lemma~\ref{lemma:polytope_property} implies that
$\mu_i(v) \succsim_v \mu_{3-i}(v)$
for every 
integer $i \in \{1,2\}$ and every vertex $v \in S_i$. 

\begin{lemma} \label{lemma:polytope_sum}
$x(E(v)) = 1$ for every vertex $v \in S_1 \cup S_2$. 
\end{lemma}
\begin{proof}
This lemma immediately follows from 
Lemmas~\ref{lemma:polytope_property} and 
\ref{lemma:polytope_size}.
\end{proof}

\begin{lemma} \label{lemma:polytope_e1}
For every integer $i \in \{1,2\}$ and 
every vertex $v \in S_i$, 
if $F_i(v) \cap E_1 \neq \emptyset$, then
we have $\mu_i(v) \in E_1$.
\end{lemma}
\begin{proof}
Let $i$ be an integer in $\{1,2\}$.
Let $v$ be a vertex in $S_i$ such that 
$F_i(v) \cap E_1 \neq \emptyset$, 
and let $e = \{v,w\}$ be an edge in 
$F_i(v) \cap E_1$. 
Lemma~\ref{lemma:polytope_size} implies that 
$\mu_i(w) \neq \emptyset$. 
Thus, 
it suffices to prove that $F_i(w) = \{e\}$. 
Lemma~\ref{lemma:polytope_property} implies that
${\sf B}(w) = \{e\}$. 
If there exists an edge $f \in F_i(w) \setminus \{e\}$, then
Lemma~\ref{lemma:polytope_property} implies that 
$f \in {\sf B}(w)$. 
This contradicts the fact that 
${\sf B}(w) = \{e\}$. 
\end{proof}

We now ready to prove Lemma~\ref{lemma:polytope_extreme_point}. 
Recall that Lemma~\ref{lemma:polytope_zero_one} implies that it is sufficient to 
prove that $x \in \{0,1\}^E$. 
For each vector $y \in \mathbb{R}_+^E$ and each edge $e \in E_1$, 
we define ${\sf LS}_1(e;y)$ by 
\begin{equation*}
{\sf LS}_1(e;y) := y(e) + \sum_{v \in e}y(E[\succ_v e]).
\end{equation*}
Define $\mathcal{E}_2$ as the set of pairs 
$(e,v)$ such that $e \in E_2$ and $v \in e$. 
For each vector $y \in \mathbb{R}_+^E$ 
and each element $(e,v) \in \mathcal{E}_2$,
we define ${\sf LS}_2(e,v;y)$ by 
\begin{equation*}
{\sf LS}_2(e,v;y) := y(E[\sim_v e]) + \sum_{w \in e}y(E[\succ_w e]). 
\end{equation*}
Define $E_1^{\ast}$ as the set of edges $e \in E_1$ 
such that ${\sf LS}_1(e;x) = 1$. 
Furthermore, we define $\mathcal{E}_2^{\ast}$ as the 
set of pairs 
$(e,v) \in \mathcal{E}_2$ such that 
${\sf LS}_2(e,v;x) = 1$. 
Define $\varepsilon_0, \varepsilon_1, \varepsilon_2$ by 
\begin{equation*}
\varepsilon_0 := 
\min_{e \in {\sf supp}}x(e), 
\ \ \ 
\varepsilon_1 := \min_{e \in E_1 \setminus E_1^{\ast}}{\sf LS}_1(e;x)-1, 
\ \ \ 
\varepsilon_2 := \min_{(e,v) \in \mathcal{E}_2 \setminus \mathcal{E}_2^{\ast}}
{\sf LS}_2(e,v;x)-1.
\end{equation*}
Define $\varepsilon := \min\{\varepsilon_0, \varepsilon_1,\varepsilon_2\}/2 > 0$.
Define 
$z_i := x + \varepsilon(\chi_{\mu_i} - \chi_{\mu_{3-i}})$
for each integer $i \in \{1,2\}$.

\begin{lemma} \label{lemma:polytope_feasibility} 
For every integer $i \in \{1,2\}$, 
$z_i \in {\bf S}$. 
\end{lemma}
\begin{proof}
Let $i$ be an integer in $\{1,2\}$.
Since $\varepsilon < \varepsilon_0$, 
we have $z_i \in \mathbb{R}_+^E$.
Thus, in what follows, we prove that 
$z_i$ satisfies \eqref{eq_1:constraint}, 
\eqref{eq_2:constraint}, and \eqref{eq_3:constraint}. 

\eqref{eq_1:constraint}
Since $x(E(v)) = z_i(E(v))$ for every vertex $v \in V$ and 
$x$ satisfies \eqref{eq_1:constraint}, 
$z_i$ satisfies \eqref{eq_1:constraint}. 

\eqref{eq_2:constraint}
Let $e = \{v,w\}$ be an edge in $E_1$
such that $e \cap V_i = \{v\}$.
If $e \in E_1 \setminus E_1^{\ast}$, then 
\begin{equation*}
{\sf LS}_1(e;z_i) 
\ge {\sf LS}_1(e;x) - 2 \varepsilon 
\ge {\sf LS}_1(e;x) - \varepsilon_1
\ge {\sf LS}_1(e;x) - ({\sf LS}_1(e;x)-1)=1.
\end{equation*}
Thus, we assume that $e \in E^{\ast}_1$. 
Since ${\sf LS}_1(e;x) \ge 1$, it is sufficient to prove that 
\begin{equation} \label{eq_1:lemma:polytope_feasibility}
{\sf LS}_1(e;z_i) \ge {\sf LS}_1(e;x).
\end{equation}

If $e \in \mu_i$, then
$e \succsim \mu_{3-i}(v)$. 
Thus, \eqref{eq_1:lemma:polytope_feasibility} holds. 
Assume that $e \in \mu_{3-i} \setminus \mu_i$. 
Then $e = \mu_{3-i}(v)$.
If $\mu_i(v) \succ_v e$, then 
\eqref{eq_1:lemma:polytope_feasibility} holds.
Thus, we assume that 
$\mu_i(v) \sim_v e$. 
Since $e \notin \mu_i$, 
$e \neq \mu_i(v)$. 
Since $\mu_{3-i}(v) \in {\sf B}(v)$, 
$\mu_i(v) \in {\sf B}(v)$. 
Lemma~\ref{lemma:polytope_property} implies that 
${\sf B}(v) = \{\mu_{3-i}(v)\}$. 
This is a contradiction. 

We consider the case where 
$e \notin \mu_1 \cup \mu_2$.
First, we assume that $\mu_i(v) \succ_v e$. 
If  
$e \succsim_v \mu_{3-i}(v)$, 
then 
\eqref{eq_1:lemma:polytope_feasibility} holds. 
If $\mu_{3-i}(v) \succ_v e$, then 
since $\mu_{3-i}(v) \in {\sf B}(v)$, 
Lemma~\ref{lemma:polytope_sum} implies that 
$x(E[\succ_v e]) = 1$. 
Thus, since $e \in E^{\ast}_1$, we have 
$x(e) = 0$ and 
$x(E[\succ_w e]) = 0$. 
This implies that 
$e \succsim_w \mu_{3-i}(w)$, and thus  
\eqref{eq_1:lemma:polytope_feasibility} holds.
Next, we assume that 
$e \succsim_v \mu_i(v)$. 
Then since $\mu_i(v) \in {\sf T}(v)$, 
$x(E[\succ_v e]) = 0$. 
Thus, we have $x(e) + x(E[\succ_w e]) = 1$. 
This implies that 
since $\mu_i(w) \in {\sf B}(w)$ and 
$e \neq \mu_i(w)$, 
we have 
$\mu_i(w) \succ_w e$. 
Thus, \eqref{eq_1:lemma:polytope_feasibility} holds.

\eqref{eq_3:constraint}
Let $e = \{v,w\}$ be an edge in $E_2$
such that $e \cap V_i = \{v\}$.

First, we consider \eqref{eq_3:constraint} for $(e,v)$. 
If $(e,v) \in \mathcal{E}_2 \setminus \mathcal{E}_2^{\ast}$, then 
\begin{equation*}
{\sf LS}_2(e,v;z_i) 
\ge {\sf LS}_2(e,v;x) - 2 \varepsilon 
\ge {\sf LS}_2(e,v;x) - \varepsilon_2
\ge {\sf LS}_2(e,v;x) - ({\sf LS}_2(e,v;x)-1)=1.
\end{equation*}
Thus, we assume that 
$(e,v) \in \mathcal{E}_2^{\ast}$. 
Since ${\sf LS}_2(e,v;x) \ge 1$, 
it is sufficient to prove that
\begin{equation} \label{eq_2:lemma:polytope_feasibility}
{\sf LS}_2(e,v;z_i) \ge {\sf LS}_2(e,v;x).
\end{equation}

If $e \in \mu_{3-i}$, then 
$\mu_i(v) \succsim_v e$, and 
thus \eqref{eq_2:lemma:polytope_feasibility} holds.
Assume that $e \in \mu_i \setminus \mu_{3-i}$. 
If $e \succ_v \mu_{3-i}(v)$, 
then \eqref{eq_2:lemma:polytope_feasibility} holds. 
Assume that $e \sim_v \mu_{3-i}(v)$. 
Then $\mu_{3-i}(v) \in {\sf T}(v)$.
Thus, since $\mu_{3-i}(v) \in {\sf B}(v)$, 
Lemma~\ref{lemma:polytope_sum} implies that 
$x(E[\sim_v e]) = 1$. 
This 
implies that 
since $(e,v) \in \mathcal{E}_2^{\ast}$, $x(E[\succ_w e]) = 0$. 
Thus, 
$e \succsim_w \mu_{3-i}(w)$. 
Since $e \neq \mu_{3-i}(w)$ follows from 
$e \neq \mu_{3-i}(v)$, 
\eqref{eq_2:lemma:polytope_feasibility} holds.

We consider the case where 
$e \notin \mu_1 \cup \mu_2$. 
First, we assume that 
$\mu_i(v) \succsim_v e$. 
If $e \succ_v \mu_{3-i}(v)$, then 
\eqref{eq_2:lemma:polytope_feasibility} holds.
If $\mu_{3-i}(v) \succsim_v e$, then 
since $\mu_{3-i}(v) \in {\sf B}(v)$, 
Lemma~\ref{lemma:polytope_sum} implies that 
$x(E[\succsim_v e]) = 1$. 
Since $(e,v) \in \mathcal{E}^{\ast}_2$,  
$x(E[\succ_w e]) = 0$. 
This implies that 
$e \succsim_w \mu_{3-i}(w)$. 
Thus, 
since $e \neq \mu_{3-i}(w)$, 
\eqref{eq_2:lemma:polytope_feasibility} holds.
Next, we assume that 
$e \succ_v \mu_i(v)$.
Since $\mu_i(v) \in {\sf T}(v)$,
we have 
$x(E[\succsim_v e]) = 0$. 
This implies that $x(E[\succ_w e]) = 1$. 
Thus, 
$\mu_i(w) \succ_w e$
and \eqref{eq_2:lemma:polytope_feasibility} holds.

Next, we consider \eqref{eq_3:constraint} for $(e,w)$. 
If $(e,w) \in \mathcal{E}_2 \setminus \mathcal{E}_2^{\ast}$,
then 
\begin{equation*}
{\sf LS}_2(e,w;z_i) 
\ge {\sf LS}_2(e,w;x) - 2 \varepsilon 
\ge {\sf LS}_2(e,w;x) - \varepsilon_2
\ge {\sf LS}_2(e,w;x) - ({\sf LS}_2(e,w;x)-1)=1.
\end{equation*}
Thus, we assume that 
$(e,w) \in \mathcal{E}_2^{\ast}$. 
Since ${\sf LS}_2(e,w;x) \ge 1$, 
it is sufficient to prove that
\begin{equation} \label{eq_3:lemma:polytope_feasibility}
{\sf LS}_2(e,w;z_i) \ge {\sf LS}_2(e,w;x).
\end{equation}

If $e \in \mu_i$, then since 
$e \succsim_v \mu_{3-i}(v)$, 
\eqref{eq_3:lemma:polytope_feasibility} holds. 
Assume that 
$e \in \mu_{3-i} \setminus \mu_i$. 
If $\mu_i(v) \succ_v \mu_{3-i}(v)$, then 
\eqref{eq_3:lemma:polytope_feasibility} holds. 
If $\mu_i(v) \sim_v \mu_{3-i}(v)$, then 
it follows from Lemma~\ref{lemma:polytope_sum} that 
$x(E[\sim_v e]) = 1$. 
Thus, we have $x(E[\succ_v e]) = 0$.
This implies that 
$x(E[\succsim_w e]) = 1$.
Thus, $\mu_i(w) \succsim_w e$
and 
\eqref{eq_3:lemma:polytope_feasibility} holds. 

We consider the case where 
$e \notin \mu_1 \cup \mu_2$.
First, we assume that 
$\mu_i(v) \succ_v e$. 
If $e \succsim_v \mu_{3-i}(v)$, then  
\eqref{eq_3:lemma:polytope_feasibility} holds. 
Thus, we assume that $\mu_{3-i}(v) \succ_v e$. 
Since $\mu_{3-i}(v) \in {\sf B}(v)$, 
Lemma~\ref{lemma:polytope_sum} implies that 
$x(E[\succ_v e]) = 1$. 
Thus, since $(e,v) \in \mathcal{E}^{\ast}_2$,  
we have $x(E[\succsim_w e]) = 0$. 
Thus,
$e \succ_w \mu_{3-i}(w)$ and 
\eqref{eq_3:lemma:polytope_feasibility} holds. 
Next, we assume that 
$e \succsim_v \mu_i(v)$.
Since $\mu_i(v) \in {\sf T}(v)$, 
$x(E[\succ_v e]) = 0$. 
This implies that $x(E[\succsim_w e]) = 1$. 
Thus, 
$\mu_i(w) \succsim_w e$ and 
\eqref{eq_3:lemma:polytope_feasibility} holds. 
\end{proof}

We are now ready to prove that 
$x \in \{0,1\}^E$. 
Since $x$ is an extreme point of ${\bf S}$, 
Lemma~\ref{lemma:polytope_feasibility} implies that 
we have $\mu_1 = \mu_2$. 
Define $\mu := \mu_1$. 
Define $W_1$ as the set of vertices $v \in S_1 \cup S_2$ such that 
$\mu(v) \in E_1$. 
Define $W_2 := (S_1 \cup S_2) \setminus W_1$. 
That is, $W_2$ is the set of vertices $v \in S_1 \cup S_2$ such that 
$\mu(v) \in E_2$. 
In what follows, we prove that
${\sf supp}(v) = \{\mu(v)\}$ for every 
vertex $v \in S_1 \cup S_2$. 
If we can prove this, then, for every vertex $v \in S_1 \cup S_2$, 
Lemma~\ref{lemma:polytope_sum} implies that 
$x(\mu(v)) = 1$ and 
$x(e) = 0$ for every edge $e \in E(v) \setminus \{\mu(v)\}$.
This completes the proof. 

\begin{claim} \label{claim:polytope_w1_property}
Let $v$ be a vertex in $W_1$. 
Then ${\sf supp}(v) = \{\mu(v)\}$.
\end{claim}
\begin{proof}
Since $\mu(v) \in {\sf T}(v)$ and 
$\mu(v) \in {\sf B}(v)$, we have 
${\sf supp}(v) \subseteq {\sf B}(v)$. 
Thus, since ${\sf B}(v) = \{\mu(v)\}$, we have 
${\sf supp}(v) = \{\mu(v)\}$. 
This completes the proof. 
\end{proof} 

For every vertex $v \in W_2$, 
Lemma~\ref{lemma:polytope_property} implies that 
$\mu(v) \in {\sf T}(v) \cap {\sf B}(v)$.
Thus, for every vertex $v \in W_2$, 
we have ${\sf supp}(v) = {\sf T}(v)$.
For every vertex $v \in W_2$, 
since $\mu(v) \in E_2$, 
Lemma~\ref{lemma:polytope_e1} implies that 
${\sf supp}(v) \subseteq E_2$. 
Let $W^{\prime}_2$ be the set of vertices $v \in W_2$ such that 
$|{\sf supp}(v)| \ge 2$. 

\begin{claim} \label{claim:polytope_w2_property}
For every vertex $v \in W_2^{\prime}$ and 
every edge $e = \{v,w\} \in {\sf supp}(v)$, 
we have $w \in W_2^{\prime}$. 
\end{claim}
\begin{proof}
First, we prove that $w \in W_2$.
If $w \in W_1$, then Claim~\ref{claim:polytope_w1_property} implies that 
${\sf supp}(w) = \{\mu(w)\}$. 
Thus, since $e \in {\sf supp}(w)$, 
$\mu(v) = e = \mu(w) \in E_1$.
This contradicts the fact that $v \in W_2$. 

Since $v \in W_2^{\prime}$, $x(e) < 1$. 
Thus, Lemma~\ref{lemma:polytope_sum} implies that 
$w \in W_2^{\prime}$. 
\end{proof} 

The following claim implies that 
${\sf supp}(v) = \{\mu(v)\}$ for every 
vertex $v \in W_2$.

\begin{claim} \label{claim:polytope_w2_empty}
$W^{\prime}_2 = \emptyset$. 
\end{claim}
\begin{proof}
Assume that $W^{\prime}_2 \neq \emptyset$. 
Claim~\ref{claim:polytope_w2_property}
implies that 
there exists a cycle $(v_1,v_2,\dots,v_{2k+1})$ 
in ${\sf supp}$.  
Notice that $v_{\ell} \in W_2^{\prime}$ and 
$\{v_{\ell},v_{\ell+1}\} \in E_2$ for every 
integer $\ell \in [2k]$. 
For each integer $\ell \in [k]$, we define 
$f_{\ell}^1 := \{v_{2\ell-1},v_{2\ell}\}$
and
$f_{\ell}^2 := \{v_{2\ell},v_{2\ell+1}\}$. 
Define 
\begin{equation*}
\delta_1 := \min_{\ell \in [k]}x(f^1_{\ell}), \ \ \ \ 
\delta_2 := \min_{\ell \in [k]}x(f^2_{\ell}), \ \ \ \ 
\delta := \min\{\delta_1, \delta_2\} > 0.
\end{equation*}
For each integer $i \in \{1,2\}$, 
we define $z_i \in \mathbb{R}_+^E$ as follows. 
\begin{equation*}
z_i(e) := 
\begin{cases}
x(e) + \delta & \mbox{if $e = f^i_{\ell}$ for some integer $\ell \in [k]$} \\
x(e) - \delta & \mbox{if $e = f^{3-i}_{\ell}$ for some integer $\ell \in [k]$} \\
x(e) & \mbox{otherwise}.
\end{cases}
\end{equation*}
If we can prove that $z_1,z_2 \in {\bf S}$, then since 
$x = (z_1 + z_2)/2$, 
this contradicts the fact that $x$ is an extreme point of ${\bf S}$.
Thus, the proof is done. 
In what follows, 
let $i$ be an integer in $\{1,2\}$, and 
we prove that 
$z_i \in {\bf S}$.

Since ${\sf supp}(v) = {\sf T}(v)$ 
for every vertex $v \in W_2$, 
$x(E[\sim_v e]) = z_i(E[\sim_v e])$ 
and 
$x(E[\succ_v e]) = z_i(E[\succ_v e])$ 
for every edge $e \in E$ and every vertex $v \in e$. Furthermore, 
since ${\sf supp}(v) \subseteq E_2$ for every vertex $v \in W_2$, 
$x(e) = z_i(e)$ for every edge $e \in E_1$. 
Thus, since $x \in {\bf S}$, $z_i \in {\bf S}$. 
\end{proof} 

\section{Structure} 
\label{section:structure} 

In this section, we first consider the set of
vertices  
covered by a non-uniformly stable matching. 
Next, we prove that 
the set of non-uniformly stable matchings forms a 
distributive lattice. 

In the proofs of the following results, we do not use 
the information about which of $E_1$ and $E_2$ 
an edge in $E$ belongs to. 
We only use the fact that, for every matching $\mu$ in $G$, 
if an edge $e \in E \setminus \mu$ strongly blocks $\mu$, then 
$e$ weakly blocks $\mu$.
Thus, the following proofs are basically 
the same as the proofs of the results 
for strong stability~\cite{Manlove99,Manlove02}. 
For completeness, 
we give the full proofs of the following results. 

\begin{theorem} \label{theorem:structure_cover}
For every pair of non-uniformly stable matchings $\mu_1,\mu_2$ in $G$, 
$V(\mu_1) = V(\mu_2)$. 
\end{theorem}
\begin{proof}
Assume that there exists a pair of non-uniformly stable matchings $\mu_1,\mu_2$ in $G$
such that $V(\mu_1) \neq V(\mu_2)$. 
Define $F := (\mu_1 \setminus \mu_2) \cup (\mu_2 \setminus \mu_1)$. 
Since $V(\mu_1) \neq V(\mu_2)$, 
there exists a vertex $v_1 \in V$ such that 
$|F(v_1)| = 1$.
Assume that $F(v_1) = \{e_1\}$. 
Since $V$ is finite, 
there exists a path $(v_1,v_2,\dots,v_k)$ in $F$
such that 
(i) $|F(v_k)| = 1$, and 
(ii) 
$|F(v_{\ell+1}) \cap \mu_1| = |F(v_{\ell+1}) \cap \mu_2| = 1$
for every integer $\ell \in [k-2]$. 
For each integer $\ell \in [k-1]$, we define 
$e_{\ell} := \{v_{\ell},v_{\ell+1}\}$. 
Since $\mu_1,\mu_2$ are non-uniformly stable, 
$e_{\ell+1} \succ_{v_{\ell+1}} e_{\ell}$ 
for every integer $\ell \in [k-2]$.
Assume that $e_{k-1} \in \mu_i$.
Then since $\mu_{3-i}(v_k) = \emptyset$ and 
$\mu_i(v_{k-1}) \succ_{v_{k-1}} \mu_{3-i}(v_{k-1})$, 
$e_{k-1}$ blocks $\mu_{3-i}$. 
However, this contradicts the fact that 
$\mu_{3-i}$ is non-uniformly stable. 
This completes the proof. 
\end{proof} 

Next, we prove that 
the set of non-uniformly stable matchings forms a 
distributive lattice. 
Define $\mathcal{S}$ as the set of non-uniformly stable 
matchings in $G$.
For each pair of elements $\mu, \sigma \in \mathcal{S}$, 
we define the subsets 
$\mu \wedge \sigma, \mu \vee \sigma \subseteq \mu \cup \sigma$ by 
\begin{equation*}
(\mu \wedge \sigma)(v) := 
\begin{cases}
\mu(v) & \mbox{if $\mu(v) \succsim_v \sigma(v)$}\\ 
\sigma(v) & \mbox{if $\sigma(v) \succ_v \mu(v)$},  
\end{cases}
\ \ \ \ \ 
(\mu \vee \sigma)(v) := 
\begin{cases}
\mu(v) & \mbox{if $\sigma(v) \succsim_v \mu(v)$}\\ 
\sigma(v) & \mbox{if $\mu(v) \succ_v \sigma(v)$}
\end{cases}
\end{equation*} 
for each vertex $v \in V_1$. 

Let $\mu,\sigma$ be elements in $\mathcal{S}$.
Define $F := (\mu \setminus \sigma) \cup (\sigma \setminus \mu)$.
Then 
Theorem~\ref{theorem:structure_cover} implies that 
there exists 
the set $\mathcal{C}(\mu,\sigma)$ of 
cycles in $F$ satisfying the following conditions. 
\begin{itemize}
\item
For every pair of elements $(v_1,v_2,\dots,v_{k+1}), (w_1,w_2,\ldots,w_{\ell+1}) \in \mathcal{C}(\mu,\sigma)$
and every pair of integers $i \in [k]$ and $j \in [\ell]$, 
we have $v_i \neq w_j$. 
\item 
For every edge $e \in F$, there exist an element $(v_1,v_2,\dots,v_{k+1}) \in \mathcal{C}(\mu,\sigma)$
and an integer $\ell \in [k]$ such that $e = \{v_{\ell},v_{\ell+1}\}$. 
\end{itemize} 
In other words, $\mathcal{C}(\mu,\sigma)$ is the set of 
cycles that exactly covers the edges in $F$. 

\begin{lemma} \label{lemma:structure_cycle} 
Let $\mu,\sigma$ be elements in $\mathcal{S}$.
Let $(v_1,v_2,\dots,v_{k+1})$ be an element in $\mathcal{C}(\mu,\sigma)$. 
Define $e_{\ell} := \{v_{\ell},v_{\ell+1}\}$ for each vertex $\ell \in [k]$ and 
$e_0 := e_k$.
Then one of the following statements holds.
\begin{description}
\item[(C1)]
$e_{\ell} \succ_{v_{\ell}} e_{\ell-1}$ for every integer $\ell \in [k]$. 
\item[(C2)]
$e_{\ell-1} \succ_{v_{\ell}} e_{\ell}$ for every integer $\ell \in [k]$. 
\item[(C3)]
$e_{\ell} \sim_{v_{\ell}} e_{\ell-1}$ for every integer $\ell \in [k]$.
\end{description}
\end{lemma}
\begin{proof}
First, we consider the case where $e_1 \succ_{v_1} e_k$. 
If there does not exist an integer $\ell \in [k]$ 
such that 
$e_{\ell-1} \succsim_{v_{\ell}} e_{\ell}$, then 
(C1) holds. 
We assume that 
there exists an integer $\ell \in [k]$ such that 
$e_{\ell-1} \succsim_{v_{\ell}} e_{\ell}$. 
Let $j$ be the minimum integer in $[k]$ such that 
$e_{j-1} \succsim_{v_{j}} e_{j}$. 
Notice that we have $j > 1$. 
Then $e_{j-1} \succ_{v_{j-1}} e_{j-2}$. 
However, $e_{j-1}$ blocks one of $\mu$ and $\sigma$. 
This is a contradiction. 

Next, we consider the case where $e_k \succ_{v_1} e_1$. 
If there does not exist an integer $\ell \in [k]$ 
such that 
$e_{\ell} \succsim_{v_{\ell}} e_{\ell-1}$, then 
(C2) holds.
Thus, we assume that 
there exists an integer $\ell \in [k]$ such that 
$e_{\ell} \succsim_{v_{\ell}} e_{\ell-1}$. 
Let $j$ be the maximum integer in $[k]$ such that 
$e_{j} \succsim_{v_j} e_{j-1}$. 
If $j = k$, then we define $e_{j+1} := e_1$. 
Then $e_{j} \succ_{v_{j+1}} e_{j+1}$. 
However, $e_j$ blocks one of $\mu$ and $\sigma$. 
This is a contradiction. 

Third, we consider the case where $e_1 \sim_{v_1} e_k$. 
If there does not exist an integer $\ell \in [k]$ such that 
$e_{\ell} \not\sim_{v_{\ell}} e_{\ell-1}$, then (C3) holds.
Thus, we assume that 
there exists an integer $\ell \in [k]$ such that 
$e_{\ell} \not\sim_{v_{\ell}} e_{\ell-1}$. 
Let $j$ be the minimum integer in $[k]$ such that 
$e_{j} \not\sim_{v_j} e_{j-1}$.
Notice that $j > 1$.  
Assume that 
$e_{j} \succ_{v_{j}} e_{j-1}$. 
Since $\mu,\sigma \in \mathcal{S}$, 
$e_k \succ_{v_k} e_{k-1}$. 
This implies that 
$e_k$ blocks one of $\mu$ and $\sigma$. 
Assume that 
$e_{j-1} \succ_{v_{j}} e_{j}$. 
Then $e_{j-1}$ 
blocks one of $\mu$ and $\sigma$. 
\end{proof}

\begin{lemma} \label{lemma:structure_matching}
For every pair of elements $\mu, \sigma \in \mathcal{S}$, 
$\mu \wedge \sigma \in \mathcal{S}$ and $\mu \vee \sigma \in \mathcal{S}$. 
\end{lemma} 
\begin{proof}
Let $\mu,\sigma$ be elements in $\mathcal{S}$.
For every symbol $\odot \in \{\wedge, \vee\}$, 
Lemma~\ref{lemma:structure_cycle} implies that
one of the following statements holds 
for each element $(v_1,v_2,\dots,v_{2k+1}) \in \mathcal{C}(\mu,\sigma)$. 
\begin{itemize}
\item
$\{v_{2\ell-1},v_{2\ell}\} \in \mu \odot \sigma$ 
and 
$\{v_{2\ell},v_{2\ell+1}\} \notin \mu \odot \sigma$ 
for every integer $\ell \in [k]$. 
\item
$\{v_{2\ell-1},v_{2\ell}\} \notin \mu \odot \sigma$ 
and 
$\{v_{2\ell},v_{2\ell+1}\} \in \mu \odot \sigma$ 
for every integer $\ell \in [k]$. 
\end{itemize}
Thus, $\mu \wedge \sigma, \mu \vee \sigma$ are matchings in $G$.

Let $\odot$ be a symbol in $\{\wedge,\vee\}$. 
We prove that $\mu \odot \sigma$ is non-uniformly stable. 
Let $e = (v,w)$ be an edge in $E \setminus (\mu \odot \sigma)$. 
If $e \in \mu \cup \sigma$, then 
there exists an element $\xi \in \{\mu,\sigma\}$ 
such that, for every vertex $u \in e$, 
$\xi(u) = (\mu \odot \sigma)(u)$. 
Thus, in this case, since $\mu,\sigma \in \mathcal{S}$, 
$e$ does not block $\mu \odot \sigma$. 
Assume that $e \notin \mu \cup \sigma$.  
If there exists an element $\xi \in \{\mu,\sigma\}$ 
such that, for every vertex $u \in e$, 
$\xi(u) = (\mu \odot \sigma)(u)$, then 
since $\mu,\sigma \in \mathcal{S}$, 
$e$ does not block $\mu \odot \sigma$.
In what follows, we assume that there does not 
exist 
such an element $\xi \in \{\mu,\sigma\}$. 

First, we consider the case where
$(\mu \odot \sigma)(v) = \mu(v)$ and 
$(\mu \odot \sigma)(w) = \sigma(w)$. 
If $\odot = \wedge$, then 
$\mu(v) \succsim_v \sigma(v)$ and 
$\mu(w) \succ_w \sigma(w)$. 
This implies that 
$(\mu \odot \sigma)(v) \succsim_v \sigma(v)$ and 
$(\mu \odot \sigma)(w) \succsim_w \sigma(w)$. 
Thus, since $\sigma \in \mathcal{S}$, 
$e$ does not block $\mu \odot \sigma$.
If $\odot = \vee$, then 
$\sigma(v) \succsim_v \mu(v)$ and 
$\sigma(w) \succ_w \mu(w)$. 
This implies that 
$(\mu \odot \sigma)(v) \succsim_v \mu(v)$ and 
$(\mu \odot \sigma)(w) \succsim_w \mu(w)$. 
Thus, since $\mu \in \mathcal{S}$, 
$e$ does not block $\mu \odot \sigma$.

Next, we consider the case where
$(\mu \odot \sigma)(v) = \sigma(v)$ and 
$(\mu \odot \sigma)(w) = \mu(w)$. 
If $\odot = \wedge$, then 
$\sigma(v) \succ_v \mu(v)$ and 
$\sigma(w) \succsim_w \mu(w)$. 
This implies that 
$(\mu \odot \sigma)(v) \succsim_v \mu(v)$ and 
$(\mu \odot \sigma)(w) \succsim_w \mu(w)$. 
Thus, since $\mu \in \mathcal{S}$, 
$e$ does not block $\mu \odot \sigma$.
If $\odot = \vee$, then 
$\mu(v) \succ_v \sigma(v)$ and 
$\mu(w) \succsim_w \sigma(w)$. 
This implies that 
$(\mu \odot \sigma)(v) \succsim_v \sigma(v)$ and 
$(\mu \odot \sigma)(w) \succsim_w \sigma(w)$. 
Thus, since $\sigma \in \mathcal{S}$, 
$e$ does not block $\mu \odot \sigma$.
\end{proof} 

The following lemma is the same as \cite[Lemma~10]{Manlove02}.
For completeness, we give its proof. 

\begin{lemma} \label{lemma:structure_operation} 
Let $\mu,\sigma,\xi$ be elements in $\mathcal{S}$, and let $v$ be a vertex in $V_1$. 
Furthermore, let $\oplus$ be a symbol in $\{\wedge, \vee\}$, and 
let $\ominus$ be the symbol in $\{\wedge, \vee\} \setminus \{\oplus\}$. 
Then the following statements hold.
\begin{gather}
(\mu \oplus \mu)(v) = \mu(v), \ \ 
(\mu \oplus \sigma)(v) \sim_v (\sigma \oplus \mu)(v), \ \ 
(\mu \oplus (\mu \ominus \sigma))(v) = \mu(v).
\label{eq_1:lemma:structure_operation} \\
(\mu \oplus (\sigma \oplus \xi))(v) = ((\mu \oplus \sigma) \oplus \xi)(v). \label{eq_2:lemma:structure_operation} \\  
(\mu \oplus (\sigma \ominus \xi))(v) = ((\mu \oplus \sigma) \ominus (\mu \oplus \xi))(v). 
\label{eq_3:lemma:structure_operation}
\end{gather}
\end{lemma}
\begin{proof}
It is not difficult to see that 
\eqref{eq_1:lemma:structure_operation}
follows from the definition. 
If $\phi_1(v) \sim_v \phi_2(v)$ holds for every 
pair of distinct elements $\phi_1,\phi_2 \in \{\mu,\sigma,\xi\}$, then 
\eqref{eq_2:lemma:structure_operation} and
\eqref{eq_3:lemma:structure_operation}
clearly hold.
In addition, if $\phi_1(v) \not\sim_v \phi_2(v)$ holds for every 
pair of distinct elements $\phi_1,\phi_2 \in \{\mu,\sigma,\xi\}$, then 
\eqref{eq_2:lemma:structure_operation} and
\eqref{eq_3:lemma:structure_operation}
clearly hold.
Thus, in what follows, we assume that 
there exist distinct 
elements $\phi_1,\phi_2,\phi_3 \in \{\mu,\sigma,\xi\}$
such that 
$\phi_1(v) \sim_v \phi_2(v) \not\sim_v \phi_3(v)$, and 
then we consider 
\eqref{eq_2:lemma:structure_operation} and
\eqref{eq_3:lemma:structure_operation}. 

First, we consider the case 
where $\oplus = \wedge$.
If $\phi_3(v) \succ_v \phi_2(v)$, then 
the left-hand side and the right-hand side 
of \eqref{eq_2:lemma:structure_operation} (resp.\ \eqref{eq_3:lemma:structure_operation})
are $\phi_3(v)$ (resp.\ $\mu(v)$). 
Thus, we assume that $\phi_2(v) \succ_v \phi_3(v)$
If $\phi_3 = \mu$, then 
the left-hand sides and the right-hand sides 
of 
\eqref{eq_2:lemma:structure_operation} and 
\eqref{eq_3:lemma:structure_operation}
are $\sigma(v)$.
Otherwise (i.e., $\phi_3 \neq \mu$), 
the left-hand sides and the right-hand sides 
of 
\eqref{eq_2:lemma:structure_operation} and 
\eqref{eq_3:lemma:structure_operation}
are $\mu(v)$.

Next, we consider the case 
where $\oplus = \vee$.
If $\phi_2(v) \succ_v \phi_3(v)$, then 
the left-hand side and the right-hand side 
of \eqref{eq_2:lemma:structure_operation} (resp.\ 
\eqref{eq_3:lemma:structure_operation})
are $\phi_3(v)$ (resp.\ $\mu(v)$). 
Thus, we assume that $\phi_3(v) \succ_v \phi_2(v)$
If $\phi_3 = \mu$, then 
the left-hand sides and the right-hand sides 
of 
\eqref{eq_2:lemma:structure_operation} and 
\eqref{eq_3:lemma:structure_operation}
are $\sigma(v)$.
Otherwise (i.e., $\phi_3 \neq \mu$), 
the left-hand sides and the right-hand sides 
of 
\eqref{eq_2:lemma:structure_operation} and 
\eqref{eq_3:lemma:structure_operation}
are $\mu(v)$.
\end{proof} 

Define the binary relation $\equiv$ on $\mathcal{S}$ as follows. 
For each pair of elements $\mu,\sigma \in \mathcal{S}$, 
$\mu \equiv \sigma$ if and only if $\mu(v) \sim_v \sigma(v)$ holds 
for every vertex $v \in V_1$.
Then it is not difficult to see that 
$\equiv$ is an equivalence relation on $\mathcal{S}$. 
Define $\mathbb{C}$ as the set of equivalence classes 
given by $\equiv$. 
Then, for each element $\mu \in \mathcal{S}$, 
let $\langle \mu \rangle$ denote the element $C \in \mathbb{C}$ such that 
$\mu \in C$. 
Define the binary operations $\wedge_{\equiv},\vee_{\equiv}$ on 
$\mathbb{C}$ as follows.
For each pair of elements $\langle \mu \rangle, \langle \sigma \rangle \in \mathbb{C}$, 
we define 
$\langle \mu \rangle \wedge_{\equiv} \langle \sigma \rangle := \langle \mu \wedge \sigma\rangle$
and 
$\langle \mu \rangle \vee_{\equiv} \langle \sigma \rangle := \langle \mu \vee \sigma\rangle$. 
Notice that 
since Lemma~\ref{lemma:structure_matching} implies that
$\mu \wedge \sigma \in \mathcal{S}$ and 
$\mu \vee \sigma \in \mathcal{S}$
for 
every pair of elements $\mu,\sigma \in \mathcal{S}$, 
the binary operations $\wedge_{\equiv}, \vee_{\equiv}$ are well-defined. 

We are now ready to give the main result of this section. 

\begin{theorem} \label{theorem:structure_main} 
$(\mathbb{C},\wedge_{\equiv},\vee_{\equiv})$ is a distributive lattice. 
\end{theorem}
\begin{proof}
Let $\mu,\sigma,\xi$ be elements in $\mathcal{S}$. 
Furthermore, let $\oplus$ be a symbol in $\{\wedge_{\equiv}, \vee_{\equiv}\}$, and 
let $\ominus$ be the symbol in $\{\wedge_{\equiv}, \vee_{\equiv}\} \setminus \{\oplus\}$. 
Then it suffices to prove that the following statements hold. 
\begin{gather}
\langle \mu \rangle \oplus \langle \mu \rangle = \langle \mu \rangle, \ \ 
\langle \mu \rangle \oplus \langle \sigma \rangle 
= \langle \sigma \rangle \oplus \langle \mu \rangle, \ \ 
\langle \mu \rangle \oplus (\langle \mu \rangle \ominus \langle \sigma \rangle) 
= \langle \mu \rangle. \label{eq_1:theorem:structure_main} \\
\langle \mu \rangle \oplus (\langle \sigma \rangle \oplus \langle \xi \rangle) 
= (\langle \mu \rangle \oplus \langle \sigma \rangle) 
\oplus \langle \xi \rangle. \label{eq_2:theorem:structure_main} \\  
\langle \mu \rangle \oplus (\langle \sigma \rangle \ominus \langle \xi \rangle) 
= (\langle \mu \rangle \oplus \langle \sigma \rangle) 
\ominus (\langle \mu \rangle \oplus \langle \xi \rangle). 
\label{eq_3:theorem:structure_main}
\end{gather}
Then 
since 
$\langle \mu_1 \rangle = \langle \mu_2 \rangle$
for every pair of elements $\mu_1,\mu_2 \in \mathcal{S}$
such that $\mu_1(v) \sim_v \mu_2(v)$ for every vertex $v \in V_1$, 
\eqref{eq_1:lemma:structure_operation} 
implies 
\eqref{eq_1:theorem:structure_main}. 
Furthermore, 
\eqref{eq_2:lemma:structure_operation} 
and 
\eqref{eq_3:lemma:structure_operation} 
imply
\eqref{eq_2:theorem:structure_main} 
and 
\eqref{eq_3:theorem:structure_main}, 
respectively.
\end{proof}

\bibliographystyle{plain}
\bibliography{non-uniform_bib}

\end{document}